\documentclass[submission,copyright,creativecommons]{eptcs}
\providecommand{\event}{FROM 2026} 

\usepackage{iftex}

\ifpdf
  \usepackage{underscore}         
  \usepackage[T1]{fontenc}        
\else
  \usepackage{breakurl}           
\fi

\usepackage{pdfx}
\usepackage{listings}
\usepackage{graphicx}
\usepackage{euscript}
\usepackage{tabularx}
\usepackage{xurl}

\usepackage{subcaption}
\usepackage{multirow}
\usepackage{amsthm}
\usepackage{amsfonts}
\usepackage{amssymb}
\usepackage{amsmath}
\usepackage{booktabs}
\usepackage{blkarray, bigstrut}

\newcommand{\cdom}{{\ensuremath{\EuScript{D}}}}

\newcommand{\Inp}[1]{I(#1)}
\newcommand{\Out}[1]{O(#1)}
\newcommand{\Inh}[1]{H(#1)}

\newcommand{\ColClasses}{\ensuremath{\Sigma}}
\newcommand{\Guards}{\ensuremath{\Phi}}

\newcommand{\mk}[1]{\ensuremath{\mathbf{#1}}}

\newcommand{\vect}[1]{\ensuremath{\mathbf{#1}}}

\newcommand{\Bag}[1]{\ensuremath{Bag[{#1}]}}
\newcommand{\Gbag}[1]{\ensuremath{Bag^*[{#1}]}}
\newcommand{\fset}[2]{\ensuremath{\{#1 \rightarrow #2\}}}

\newcommand{\p}[1]{\ensuremath{\mathrm{#1}}}

\newcommand{\SNex}{\textsf{SNexpression}}
\newcommand{\GreatSPN}{\textsf{GreatSPN}}

\newcommand{\tuple}[1]{\langle{#1}\rangle}
\newcommand{\lang}{\EuScript{L}}
\newcommand{\Nat}{\mathbb{N}} 
\newcommand{\Int}{\mathbb{Z}}

\newtheorem{definition}{Definition}[section]

\newtheorem{property}{Property}[section]

\newtheorem{proposition}{Proposition}[section]

\newtheorem{claim}{Claim}[section]

\begin{document}

\title{Semi-automated Verification of Symbolic  Invariants In Extended Symmetric Nets}



\author{Lorenzo Capra
\institute{
{Department of Informatics}\\
    {Universit{\`a} degli Studi di Milano}, Italy}}
    
\def\titlerunning{Symbolic Verification of Invariants in ESN}
\def\authorrunning{L. Capra}



\maketitle

\begin{abstract}
Structural analysis is central to Petri Net (PN) research, complementing state‑space methods while avoiding their combinatorial issues. It is well studied for classical PNs but much less for High‑Level Petri Nets (HLPN). Symmetric Nets (SN), a common HLPN formalism, use compact annotations to encode behavioral symmetries, allowing symbolic reachability graphs and associated lumped Markov chains for stochastic SN.
In the past two decades, specific structural techniques for SN have emerged, notably the \SNex\ tool, which implements a calculus for symbolic structural relations such as conflict and causality. We propose using this calculus to semi-‐automatically verify symbolic structural invariants, currently possible only for restricted SN subclasses, for an \emph{extended} SN formalism (ESN) closed under key functional operators. We focus on flows and outline, at least in theory, how to construct a flow‑generating family. We also sketch a framework for formally verifying a wider range of invariant properties. Representative examples illustrate the main concepts.

\end{abstract}



\section{Introduction And Related Work}
\label{intro}
High-Level Petri Nets (HLPNs) \cite{HLPN} extend basic Petri nets (PNs) \cite{ReisigPN} by allowing structured tokens and typed places and transitions, so each node may have multiple instantiations. Colored Petri Nets (CPNs) \cite{Jensen1991,CPN09} are HLPNs where nodes are associated with finite color sets and arcs map transition colors to multisets over place colors. Any CPN can be unfolded into an equivalent Place/Transition (PT) net \cite{Jensen1991,ReisigPN}, but for realistic models, these unfoldings are often prohibitively large or infeasible to construct, and they hinder relating analysis results back to the CPN.

Symmetric Nets (SN) \cite{CDFH93}, formerly Well-formed Nets, form a CPN subclass whose structured syntax captures behavioral symmetries. Node domains are Cartesian products of basic classes, partitionable into subclasses gathering behaviorally similar entities, and arcs are labeled with tuples of base color functions in these domains. Despite these restrictions, SNs are as expressive as standard CPNs \cite{Jensen1991}. Using symbolic initial markings and firing rules, one can construct a compact Symbolic Reachability Graph strongly bisimilar to the ordinary one, which for Stochastic SNs yields a lumped Markov chain \cite{GreatSPN}.

A key advantage of PNs is that meaningful properties can be derived directly from their structure. Extending this to CPNs symbolically (i.e. without unfolding) is a challenge. The structural analysis of SN has advanced considerably in the past two decades. Its theoretical foundations, established in \cite{CAPRA2005}, define a language to symbolically express and compute key structural relations (conflict, causal connection, mutual exclusion) \cite{QEST2015}. This language resembles SN arc functions but is a little more expressive. 
The {\SNex} software tool \cite{Capra2020} (\url{http://www.di.unito.it/~depierro/SNexpression}), a CLI built on an extensible Java library, is a computer algebra system that reduces user-specified structural expressions over SN node pairs to normal forms. A recent version \cite{FORTE25} extends this calculus to the matrix level (the whole SN), enabling the computation of more general structural dependencies, including priorities.

This paper shows that the symbolic calculus underlying \SNex\ can be exploited to verify structural \emph{invariants} in an \emph{extension} of SN (ESN) with favorable algebraic properties, focusing on (semi)flows. The ESN arc functions form a closed class under key functional operators, in particular composition \cite{pnse2021}. To our knowledge, no other method computes and verifies symbolic invariants without constraining SN color functions.
This complements the computation of structural dependencies already available.

First, we consider the verification of symbolic (semi)flows, currently available only for restricted subclasses of SN models. Second, we introduce colored flows, which ignore token multiplicities and focus on color types, without relying on conservative assumptions. Third, we sketch how to derive a generative basis of semiflows by exploiting the constant-size property of ESN functions and relating symbolic semiflows to conventional semiflows of the underlying P/T-net skeleton. Fourth, we outline a more general approach based on standard inductive inference techniques, whose automation requires extending the supporting tool. 

All concepts are illustrated through representative examples. In this study, we automated the majority of the elementary steps required by the first three methods; in principle, these methods are fully amenable to automation in \SNex. By contrast, the fourth technique would require substantial modifications to enable the dynamic introduction of new rules into the underlying rewriting engine.

\vspace{-7pt}
\paragraph{Related work}
Generative flow families for classical Petri nets (PNs)—integer solutions of a homogeneous linear system with the incidence matrix as the coefficient matrix—have been thoroughly studied and successfully solved since \cite{colom1989}. Extending these techniques to High-Level Petri Nets (HLPN), especially for automated verification and symbolic flow generation, has also been widely investigated, but substantial results exist only for restricted HLPN subclasses, mostly using symbolic variants of classical elimination procedures (e.g., Gaussian or Gauss–Jordan elimination).
In \cite{Couvreur1991}, an algorithm is given for generative flow families in Commutative High Level Nets, a CPN subclass whose color functions form a ring of commutative endomorphisms. This work was later extended to Unary Regular Nets and unary Predicate/Transition nets in \cite{CouvHad93}. The method in \cite{Silva1991}, based on generalized inverses of arc functions, is in principle applicable to general CPNs but has never been automated because symbolic computation of these inverses is practically infeasible. Since then, only a few major theoretical advances have appeared; among them, \cite{Eva2007} proposes a block-matrix–based algorithm for a subclass of Symmetric Nets (SN), Simple Well-formed Nets. However, strict constraints on color functions, e.g. the absence of guards, make this technique unusable in many settings, including most examples in this paper.

The paper is organized as follows. Section \ref{sec:SNdef} introduces the ESN syntax. Section \ref{sec:eq-tr} establishes the fundamental properties of the language of color functions. Section \ref{sec:semiflow-def} presents symbolic flows and their basic properties, while Section \ref{sec:semiflow-ver} demonstrates flow verification through examples that highlight key features of the ESN syntax. Section \ref{sec:basis} briefly discusses the calculation of a flow basis for ESN. Section \ref{sec:cflow} defines a class of colored flows that ignore color multiplicities. Section \ref{sec:genInv} treats the verification of more general structural invariants for ESN. We conclude by indicating directions for future work.

\begin{figure}[!h]
    \centering
    \includegraphics[width=0.8\linewidth]{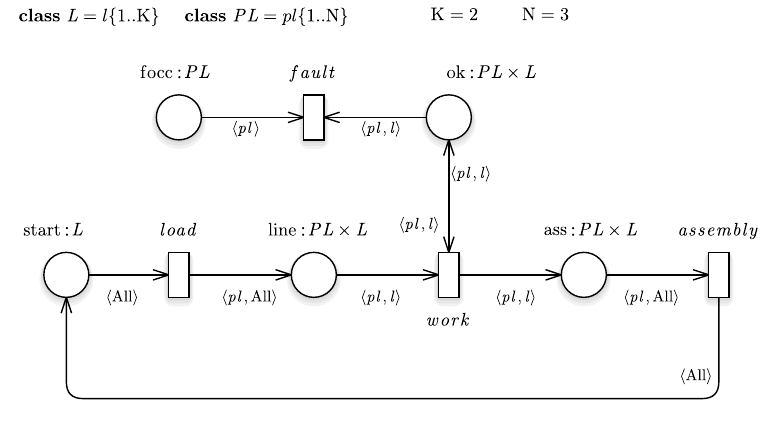}
    \caption{SN model of a distributed Production Line}
    \label{fig:PL}
\end{figure}

\section{(Extended) Symmetric Nets}
\label{sec:SNdef}

Assuming only basic HLPN knowledge, we first give an informal overview of SN syntax using an example that will reappear later. We then formally present an extended version of SN, called ESN, focusing on the components needed for the invariant calculus. Stochastic parameters are omitted, as they are irrelevant here.

The domains of (E)SN nodes are Cartesian products of finite \emph{color classes}. These classes may be partitioned into static subclasses or, alternatively, ordered circularly. Arc inscriptions are sums of tuples of basic color functions, such as projections and constants denoting subclasses or entire classes.

Consider the SN shown in Fig. \ref{fig:PL}; it refers to the core component of a distributed production system that gracefully degrades and is used as a reference in \cite{CAPRA-TCS2024,ICLP24}.
The system has $N$ parallel replicas of a production line (PL), each with $K$ mutually interchangeable components subject to faults and controlled degradation (we disregard the latter aspect here).
The PT net representing the unfolding of this SN (for $K = 2, N = 2$) is given in the appendix (Fig. \ref{fig:unfolding}).

The SN model therefore employs two unpartitioned and unordered color classes, $L$ and $PL$, of sizes $K$ and $N$, respectively (these are the model parameters).  
The places `\p{line}', `\p{ass}', and `\p{ok}' have domain $PL \times L$; the transition $work$ has the same inferred domain and its incident arcs use the tuple $\tuple{pl,l}$. Thus, an instance $\tuple{pl=c,l=c'}$ fires exactly when the required tokens $\tuple{c,c'}$ are present, moving one of such tokens from `\p{line}' to `\p{ass}'.
The transition $load$ has domain $PL$. Its input arc from `\p{start}' is labeled $\tuple{All}$, while its output arc to `\p{line}' is labeled $\tuple{pl,All}$. Hence, one firing consumes all colors of $L$ from `\p{start}' and produces the multiset $\sum_{c'\in L}\tuple{c,c'}$ in `\p{line}'. The variable $pl$ is free in the input condition, so all instances enabled by $load$ are mutually in conflict. Ordered/partitioned classes, guards, algebraic sums, and the ESN-only construct are used in later examples.

\subsection{Multiset functions and operations}
Let $C$ be a non-empty set. A multiset in $C$ is a map $\vect{b}:C\rightarrow\Int$, 
where $\vect{b}(c)$ is the multiplicity of $c$;
its ``size'' is $|\vect{b}| := \sum_{c \in C} \vect{b}(c)$ (can be negative). $\Gbag{C}$ denotes (the set of) multisets and $\Bag{C}\subset\Gbag{C}$ those with natural multiplicities. If $\vect{b}\in\Bag{C}$ then its support is $\overline{\vect{b}}:=\{c\in C\mid \vect{b}(c) \neq 0\}$.
The empty multiset is denoted by $\vect{0}$.

Multiset operations are defined pointwise. Let $\vect{a}, \vect{b}\in\Gbag{C}$. For $\lambda\in\Int$, $\lambda\vect{a}(c)=\lambda *\vect{a}(c)$; binary operations $\mu$ 
move pointwise to $\vect{a}\,\mu\,\vect{b}$. 
For example $\vect{a} +\vect{b} (c) = \vect{a}(c) + \vect{b} (c) \ \forall c \in C  $; $\vect{a} -\vect{b} \equiv \vect{a} + (-1)\vect{b}$. The operators $\cap$ ($\equiv \ min$) and $\ominus$ are used in natural multisets, where $\vect{a}\ominus \vect{b}:=min(\vect{a}-\vect{b},\vect{0})$. The Cartesian product is also lifted: if $\vect{a}'\in\Gbag{C'}$, then $\tuple{\vect{a},\vect{a}'}\in\Gbag{C\times C'}$ and $\tuple{\vect{a},\vect{a}'}(c,c') := \vect{a}(c) * \vect{a}'(c')$. The relational operators are extended in an analogous way; for instance, $\vect{a} < \vect{b} := \bigwedge_{c \in C} \vect{a}(c) < \vect{b}(c)$.
Associative operations ($+(-), \cap, \tuple{\ldots}$) can be straightforwardly generalized to their n-ary counterparts.

Functional operators are obtained by evaluation. If $f,h\in\fset{A}{\Gbag{C}}$, $f'\in\fset{A}{\Gbag{C'}}$, and $\lambda\in\Int$, then $\lambda f$, $f+h$, $\tuple{f,f'}$, comparisons, $\cap$, $\ominus$, and support are defined pointwise.
For example, $f + h \in\fset{A}{\Gbag{C}} := f(a) + h(a)$, $\forall a \in A$. The product $\tuple{f,f'} \in \fset{A}{\Gbag{C\times C'}}$ (called a function tuple) is $\tuple{f,f'}(a) := \tuple{f(a),f'(a)}$, $\forall a$. Transposition is central: for $f\in\fset{A}{\Gbag{C}}$, $f^t\in\fset{C}{\Gbag{A}}$ is $f^t(c)(a) := f(a)(c)$, $\forall a, c$.

Functions are linearly extended to multiset arguments. If $f\in\fset{A}{\Gbag{B}}$, then, abusing the notation $f(\vect{a}):=\sum_{a\in A}\vect{a}(a)f(a)$; consequently, if $g \in\fset{B}{\Gbag{C}}$ then $g\circ f$ is defined by applying the linear extension of $g$ to $f(a)$. Guards between brackets are interpreted as functions $[p]\in\fset{A}{\Gbag{A}}$ that return $a$ ($1a$) when $p(a)$ is satisfied and $\vect{0}$ otherwise. For $f\in\fset{A}{\Gbag{B}}$, suffix and prefix guards are $f[p]:=f\circ[p]$ and $[p']f:=[p']\circ f$; prefix guards play a crucial role because they act as filters.

\subsection{Extended SN: formal definition}
\label{subsec:ESN}
\newcommand{\pri}{\ensuremath \pi}
\begin{definition}[Extended SN]
An ESN is a tuple
$$ {\EuScript N}=(P,T,\ColClasses,\cdom,I,O,H,\Guards,\mk{m}_0)$$
where $P$ and $T$ are finite, nonempty, disjoint sets of places and transitions; $\ColClasses$ is a nonempty set of color classes; $\cdom$ assigns a color domain to each node; $I,O,H$ are input, output and inhibitor arc-function families; \Guards\ assigns a guard to each transition; and $\mk{m}_0$ is the initial marking.
\end{definition}

\vspace{5pt}
Each color class $C_i$ is finite and disjoint from the others; it may be partitioned into static subclasses or ordered circularly, in which case $!^m(c)$ denotes the $m$-th successor/predecessor modulo $|C_i|$. A color domain is a product $\bigotimes_i C_i^{e_i}$, and $\cdom_\ColClasses\leq\cdom'_\ColClasses$ iff each multiplicity $e_i$ is not larger than the corresponding $e_i'$. Place domains define token shapes; transition domains define transition instances, and are usually inferred from incident inscriptions. A guard restricts $\cdom(t)$ to instances that meet a standard predicate. For simplicity, we denote this restriction by $\cdom(t)$.

For $W\in\{I,O,H\}$, an arc function $W(p,t)\in\fset{\cdom(t)}{Bag[\cdom(p)]}$ belongs to the language:
\begin{equation}
\label{eq:arcfun}
 \lang_{\cdom_{\ColClasses},\cdom'_{\ColClasses}} :=
 \{ F := \sum_i \lambda_i . [g_i'] T_i [g_i],\,\lambda_i \in \Int \}
 \cup \{0_{\cdom_{\ColClasses},\cdom'_{\ColClasses}}\}.
\end{equation}
Here, $T_i=\langle f_1,\ldots,f_{|\cdom'_{\ColClasses}|}\rangle$ is a function tuple and $g_i,g'_i$ are standard predicates on the source and target domains. Each class function has the form
\begin{equation}
\label{eq:classfun}
f \in \fset{\cdom_{\ColClasses}}{\Gbag{C_h}} := \sum_{k} \alpha_k. \epsilon_k, \,\, \alpha_k \in \Int,
\end{equation}
where $\epsilon_k$ is a projection $x_h^j$ (possibly shifted by $!^m$ on ordered classes), a constant $All_h$, or a subclass constant $All_{h,q}$. Standard predicates use equality/disequality between projections, ordered variants, and subclass membership.

Negative terms may appear syntactically, following the semantics of algebraic sums, but arc functions are required to produce multisets in $\Bag{\cdom(p)}$.

\paragraph{Semantics}
A transition instance $(t,c)$ is enabled in the marking $\mk{m}$ iff \footnote{The gray part corresponds to inhibitor-arc functions; these are not included in some CPN definitions and do not affect the flow computation. They will, however, be considered in the approach outlined in Section \ref{sec:genInv}.}
$$
\forall p\in P \  \Inp{p,t}(c) \leq \mk{m}(p) \textcolor{gray}{\ \wedge \
(\Inh{p,t}(c)=0 \vee \Inh{p,t}(c)>\mk{m}(p))}.
$$
If enabled, it leads to $\mk{m}'=\mk{m}+\Out{p,t}(c)-\Inp{p,t}(c)$, written $\mk{m}[(t,c)>\mk{m}'$.

\paragraph{Conventions}
Capital words, e.g., $C,L,PL$, denote color classes, while lowercase variants denote projections; subscripts distinguish repeated classes in a domain, a capability that is not admitted in related approaches.

\paragraph{ESN arc-function example}
The following arc function
$$\tuple{All - n, n, ack} + [n_1 \neq n_2]\, \tuple{All - n,All -  n, ack} $$
is associated with the input arc that connects the place 'delivered' to the transition $complete$ of the ESN shown in Fig. \ref{fig:broadcast}. It has type $\fset{N}{N^2\times L}$ and cannot be represented in the original SN syntax. With $L_1=\{data\}$ and $L_2=\{ack\}$, it collects all the expected acknowledgments for a sender $n$, including the filtered pairs of distinct nodes in $(All-n)\times(All-n)$. The calculation of its image size is nontrivial and uses chromatic-polynomial techniques \cite{pnse2021}.

\section{Base Properties Of Arc Functions}
\label{sec:eq-tr}
The language $\lang$ is designed to be stable under the operations required by symbolic structural analysis.

\vspace{5pt}
\begin{proposition}
$\lang$ is closed under composition $(\circ)$, transposition $({}^t)$, algebraic sum $(+(-))$, difference $(\ominus)$, and intersection $(\cap)$.
\end{proposition}
The same holds for the derived language $\overline{\lang}:=\{\overline{F}\mid F\in\lang\}$ of set-valued functions, where multiplicities are ignored. Structural dependencies among SN nodes can be expressed in this language with natural coefficients and at most one difference symbol.

For invariant checking, the arc functions are assumed to have a constant size.
\begin{definition}
\label{def:c-size}
$F \in \fset{A}{\Gbag{B}}$ is constant-size if $\exists k \in \Int \, \forall a, |F(a)| = k$.
\end{definition}

\begin{proposition}
An SN transition $t$ can be split into $t'_1,\ldots,t'_k$ with uniquely constant-size arc functions such that $\cdom(t)=\bigcup_i\cdom(t'_i)$, the domains $\cdom(t'_i)$ are pairwise disjoint and $\mk{m}[(t,c)>\mk{m}'$ iff $\mk{m}[(t'_i,c)>\mk{m}'$ for some $i$.
\end{proposition}

\noindent Fig. \ref{fig:tequiv} illustrates the transformation: a transition with a size-dependent output is split so that each resulting arc function has a constant size.

\begin{figure}
    \centering
    \includegraphics[width=0.7\linewidth]{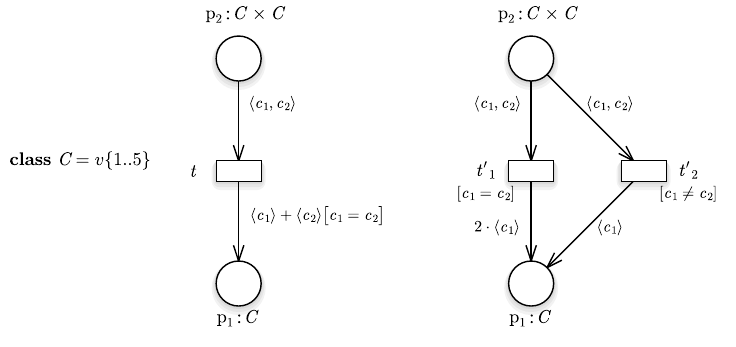}
    \caption{Transformation into constant-size functions}
    \label{fig:tequiv}
\end{figure}

Moreover, any $W\in\lang$ can be written as $\sum_i\lambda_iF_i$, where $F_i$ are pairwise disjoint guarded tuples producing multisets with unit multiplicities. This form also supports codomain checks by inspecting the coefficients $\lambda_i$. Although expressions in $\lang$ generally lack a unique canonical representative, the equivalence of $W,W'\in\lang$ is decidable by reducing both $W-W'$ and $W'-W$: the calculus always reduces null expressions to the syntactic form $\vect{0}$.


\section{Symbolic Flows: Basic Definitions And Properties}
\label{sec:semiflow-def}
We suppose that functions are constant-size. Let $\vect{H}$ be the $|P|\times |T|$ incidence matrix of an ESN, with $\vect{H}[p,t]=\Out{p,t}-\Inp{p,t}$. Matrix products use function composition.

\begin{definition}[P-(semi)flow]
Let $\vect{I}_P$ be a non-null vector $1\times |P|$ such that $\vect{I}_P[p]\in\fset{\cdom(p)}{\Gbag{\cdom'}}$ and $\vect{I}_P[p]\not\equiv0\Rightarrow\cdom'\leq\cdom(p)$. It is a P-flow iff $\vect{I}_P\cdot\vect{H}\equiv\vect{0}$. It is a P-semiflow when all non-null components map to $\Bag{\cdom'}$.
\label{def:P-semif}
\end{definition}

A minimal (semi)flow cannot be obtained as a linear combination of others. P-semiflows have the usual invariant interpretation: for every reachable marking $\vect{m}$, $\vect{I}_P\cdot\vect{m}=\vect{I}_P\cdot\vect{m}_0$. General flows may also be informative, but the examples in the following mostly use semiflows.

\begin{definition}[T-semiflow]
Let $\vect{I}_T$ be a non-null vector $|T|\times1$ such that $\vect{I}_T[t]\in\fset{\cdom'}{\Bag{\cdom(t)}}$ and $\vect{I}_T[t]\not\equiv0\Rightarrow\cdom'\leq\cdom(t)$. It is a T-semiflow iff $\vect{H}\cdot\vect{I}_T\equiv\vect{0}$.
\label{def:T-semif}
\end{definition}

\noindent A T-semiflow is a symbolic firing cycle: any sequence of transition instances conforming to an element of the image of $\vect{I}_T$ leads back to the starting marking.

These definitions differ in subtle ways from earlier formulations (e.g. \cite{Eva2007}), which construct common superdomains by inserting auxiliary “dummy" functions $All$, thus producing a model that is weakly bisimilar.
Using the largest subdomains makes it easier to interpret symbolic flows and relate them to the conventional flows of the unfolding.

\begin{definition}
\label{def:skel}
The skeleton of an ESN $\EuScript{N}=(P,T,I,O,\ldots)$ is the PT net $(P,T,F)$ with $F(p,t)=size(\Inp{p,t})$ and $F(t,p)=size(\Out{p,t})$.
\end{definition}

\begin{property}
\label{propr:skel}
If $\vect{I}_P$ is a symbolic P-flow for $\EuScript{N}$, then the integer vector $\vect{I}'_P[p]=size(\vect{I}_P[p])$ is a conventional P-flow for the skeleton $\EuScript{N}_{skel}$.
\end{property}

The analogous result holds for T-flows. Hence skeleton semiflows provide necessary candidates and substantially reduce the symbolic search space; Section \ref{sec:basis} uses this observation.

\begin{property}
\label{propr:eqT-P-flows}
Let $\vect{I}_P$ be a non-null $1\times |P|$ vector. Then
$\vect{I}_P \cdot \vect{H} \equiv \vect{0}
\Leftrightarrow  \vect{H}^t \cdot \vect{I}_P^t \equiv \vect{0}.$
\end{property}

Thus, P-flow verification can be carried out by ordinary symbolic row-by-column products or, equivalently, through the transposed matrix. This also relates the P-flows of a net to the T-flows of the dual net, obtained by swapping places and transitions and transposing arc functions.

\section{Semi-Flows Verification: A Few Examples}
\label{sec:semiflow-ver}

In this section, we validate candidate P and T-semiflows using several examples. Each of these examples outlines specific aspects of the ESN syntax. 
Two are drawn from the literature, while the others have been constructed for illustrative purposes. The calculations were carried out semi-automatically with the support of \SNex. In particular, the current implementation still does not natively support the composition in $\lang$. Nevertheless, one can simulate this composition via that in $\overline{\lang}$, provided the arc functions are (re)expressed as pairwise disjoint linear combinations $\sum_i \lambda_i F_i$ where terms $F_i$ are guarded tuples that map to multisets in which all multiplicities are equal to one (which is the natural choice of modelers). 
The individual calculations illustrated below took a few milliseconds on an 11th Gen Core i5 with 32 GB of RAM.
All models were edited using the GreatSPN package \cite{GreatSPN}, which also supports the unfolding of the SN and the computation of conventional (numerical) semiflows. A translator from the \GreatSPN\ format (\url{https://github.com/greatspn/SOURCES}) to the \SNex\ format can be downloaded from the  \SNex\ homepage.

\paragraph{Example 1}

\begin{figure}[!ht]
\[
\begin{blockarray}{ccccc}
 & load_{PL} & work_{PL,L} & assembly_{PL} & fault_{PL,L} \\
\begin{block}{c[cccc]}
\p{start} & -\tuple{All_L} &  & \tuple{All_L} & \bigstrut[t] \\
\p{line} & \tuple{pl, All_L} & -\tuple{pl, l} & &  \\
\p{ass} &  & \tuple{pl, l} & -\tuple{pl, All_L} & \\
\p{ok} & & & & -\tuple{pl, l}\\
\p{focc} & & & & -\tuple{pl}\bigstrut[b]\\
\end{block}
\end{blockarray}\vspace*{-1.25\baselineskip}
\]

\[
\begin{blockarray}{cccccc}
& \p{start} & \p{line} & \p{ass} & \p{ok} & \p{focc} \\
\begin{block}{c[ccccc]}
\vect{I}^1_P  & \tuple{l}_{L} & \tuple{l}_{PL,L} & \tuple{l}_{PL,L} &  &  \\
\end{block}
\begin{block}{c[ccccc]}
\vect{I}^{2}_P & &  & & -\tuple{pl}_{PL,L} & \tuple{pl}_{PL} \\
\end{block}
\end{blockarray}\vspace*{-1.3\baselineskip}
\]

\vspace*{0.35\baselineskip}
\[
\begin{blockarray}{ccccc}
& load_{PL} & work_{PL,L} & assembly_{PL} & fault_{PL,L} \\
\begin{block}{c[cccc]}
\vect{I}_T  & \tuple{pl}_{PL} & \tuple{pl, All_L}_{PL} & \tuple{pl}_{PL} &    \\
\end{block}
\end{blockarray}\vspace*{-1.3\baselineskip}
\]

\caption{Incidence Matrix and P- and T-semiflows of SN in Fig. \ref{fig:PL}}
\label{fig:PLflows}
\end{figure}

Fig. \ref{fig:PLflows} presents the  incidence matrix of the SN representing a distributed production line (Fig. \ref{fig:PL}) together with potential (semi)flows, according to property \ref{propr:skel} (with the T-vector also displayed as a row hereinafter); 
The SN skeleton is shown in the appendix (Fig. \ref{fig:skeleton}).

These candidate symbolic flows can be efficiently validated using the command-line interface (CLI) of \SNex. The color domains associated with the transitions, as well as the domains of the functions, are explicitly represented. It should be noted that the same symbol (e.g., $\tuple{l}$) can denote functions defined in different domains.

Semiflow $\vect{I}^1_P$ admits a straightforward interpretation. It indicates that the distributed PLs (the lower part of the figure) are conservative with respect to the color class $L$: for example, if the initial marking contains a multiple of $K$ distinct tokens (representing workpieces) in place $\p{start}$, this quantity is preserved throughout execution. 
$\vect{I}^2_P$ is the only instance of a flow that does not match a semiflow reported in the paper. This implies that the $PL$ component associated with the place $\p{ok}$ decreases consistently with the marking of the place $\p{focc}$.
Because the SN places are covered by (semi)flows, the SN is structurally bounded (and color-safe, provided that the initial marking is color-safe). 

Even more interestingly, the T-semiflow identifies a production \emph{cycle} within a given PL, represented by the symbol $pl$: this cycle consists of an occurrence of $load$, $K$ occurrences of $work$, and an occurrence of $assembly$.
Observe that these symbolic semiflows are intrinsically parametric, since no multiplicities appear either in $\vect{H}$ or in the semiflow expressions.

We now briefly analyze the relationship between symbolic and conventional (semi)flows, with reference to the unfolding of Fig. \ref{fig:PL} reported in the Appendix (see Fig. \ref{fig:unfolding}). 
Consider first $\vect{I}^1_P$. We note that the components associated with the places $\p{line}$ and $\p{ass}$ correspond to the projection onto class $L$ of the color domain $PL \times L$. Consequently, for a given color of $L$ bound to $l$, there exist as many conventional semiflows as there are possible combinations of $PL$-colors in these two places. This yields a total of $N^2 \cdot K$ conventional semiflows that instantiate $\vect{I}^1_P$. 
For example, when $N = 2$ and $K = 2$, we obtain the eight P-semiflows shown in the Appendix (Fig. \ref{fig:unfolding}).

The symbolic T-semiflow, in turn, corresponds to \(N\) conventional semiflows, each associated with a distinct binding of \(pl\) (see Fig.~\ref{fig:unfolding}).

\begin{figure}
    \centering
    \includegraphics[width=\linewidth]{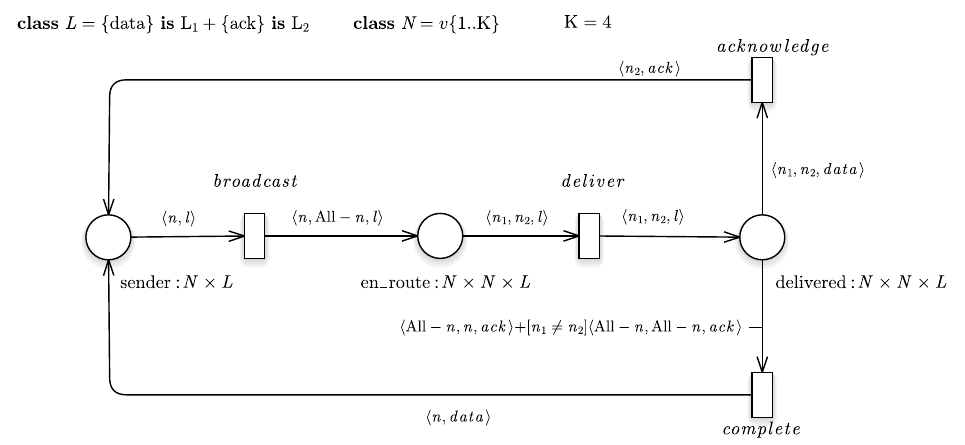}
    \caption{An ESN representing a broadcast protocol ($ack \equiv All_{L2}$, $data \equiv All_{L1}$)}
    \label{fig:broadcast}

\[
\begin{blockarray}{ccccc}
 & broadcast_{N,L} & deliver_{N,N,L} & acknowledge_{N,N} & complete_{N} \\
\begin{block}{c[cccc]}
\p{sender} & -\tuple{n,l} &  & \tuple{n_2,ack} & \tuple{n,data} \bigstrut[t] \\
\p{en\_route} & \tuple{n, All - n, l} & -\tuple{n_1, n_2, l} & &  \\
\p{delivered} &  & \tuple{n_1, n_2, l} & -\tuple{n_1, n_2, data} & -\tuple{All - n, n, ack} \\
 &  &  &  & - [n_1 \neq n_2]\tuple{All - n,All -  n, ack}\\
\end{block}
\end{blockarray}\vspace*{-0.75\baselineskip}
\]

\vspace*{0.25\baselineskip}
\[
\begin{blockarray}{ccccc}
& broadcast & deliver & acknowledge & complete \\
\begin{block}{c[cccc]}
  & (\tuple{n,data}+ &  (\tuple{n,All - n, data} + & \tuple{n,All - n}_{N} &  \tuple{n}_{N}  \\
\vect{I}_T & \tuple{All - n,ack})_{N} & \tuple{All - n, n, ack} + &  &    \\
 &  & [n_1 \neq n_2]\tuple{All - n,All -  n, ack})_{N} &  &    \\
\end{block}
\end{blockarray}\vspace*{-0.3\baselineskip}
\]

\caption{Incidence Matrix and T-semiflow of SN in Fig. \ref{fig:broadcast}}
\label{fig:T-flowbroad}
\end{figure}

\paragraph{Example 2}
The second example in this section is a variant of one presented in \cite{Camilli2021}, and it models a broadcast communication protocol. This model requires the use of the extended SN syntax. The ESN shown in Fig. \ref{fig:broadcast} (and the following) uses more sophisticated functions. The class $N$ denotes the nodes in a network, while the class $L$ is divided into two singleton subclasses that encode \emph{data} and \emph{acknowledge} messages, respectively. The protocol operates as follows: when a message $l$ is to be transmitted from a source node $n$ (the place $\p{sender}$), it is sent to all other nodes in the network (transition $broadcast$). Once a \emph{data} message has been delivered to a node, that node schedules an acknowledgment using the same broadcast mechanism. The protocol terminates successfully when all nodes, including the original sender, have received an acknowledgment.

A notable component of the SN is the function $\Inp{\p{delivered}, complete}$, previously described, which represents a sophisticated synchronization condition: all other nodes have given an acknowledgment to node $n$ (the original sender), and each of these nodes has received an acknowledgment from all nodes except $n$ and itself. 
 Within this function, a filter is applied to extract all pairs of distinct nodes from the parametric Cartesian product \((All - n) \times (All - n)\).
 
The ESN skeleton is also shown in the Appendix. Since it does not allow any P-semiflows, the ESN consequently has no symbolic P-semiflows either.
In contrast, it admits the following conventional T-semiflow (with $K = |N|$): 
\[
[broadcast: K,\; deliver: K \cdot (K - 1),\; acknowledge: K - 1,\; complete: 1].
\]
The candidate symbolic semiflow shown in Fig. \ref{fig:T-flowbroad} is in agreement with this. To confirm that it is indeed a T-semiflow, we must resolve nontrivial compositions like the one below, which arises when multiplying the second row of \vect{H} by $\vect{I}_T$. The symbolic calculus reduces all expressions $e$ to a normal form $\hat{e} \in \EuScript{L}$.

The equivalence of terms can be established purely syntactically since the calculus guarantees that, whenever $e \equiv 0$, it follows that $\hat{e} = 0$.

\vspace{-10pt}
\begin{align*}
\tuple{n, All - n, l}  \circ \tuple{All - n,ack} &\equiv \\
\tuple{n, All, l}  \circ \tuple{All - n,ack} - \tuple{n, n, l}  \circ \tuple{All - n,ack} &\equiv \\
\tuple{All - n, All, ack}  - [n_1 == n_2]\, \tuple{All - n, All - n, ack} &\equiv \\
[n_1 \neq n_2] \, \tuple{All - n, All - n, ack}  + \tuple{All - n, n, ack} & \\
\end{align*}

\vspace{-5pt}
We can directly verify that the row-by-column products are zero and therefore $\vect{I}_T$ is a semiflow. It corresponds to the completion of a basic protocol cycle: for any $n$ representing the original sender, and ignoring the specific order, it includes one $broadcast$ of a $data$ message from $n$, and $broadcast$s of an $ack$ message from all the other nodes; a set of $deliver$ events of the $data$ message from $n$ to every other node, a set of $deliver$ events of an $ack$ from each other node to $n$, and a set of $deliver$ events of an $ack$ from any node other than $n$ to any other node distinct from both $n$ and itself; the emission of an $acknowledge$ by every node other than $n$; and finally a $complete$ event, which represents the synchronization upon reception of all expected $ack$ messages by $n$ and by all other nodes.
We can similarly prove that the other row-by-column products are nullified.

\begin{figure}[t!]
    \centering
    \includegraphics[width=0.7\linewidth]{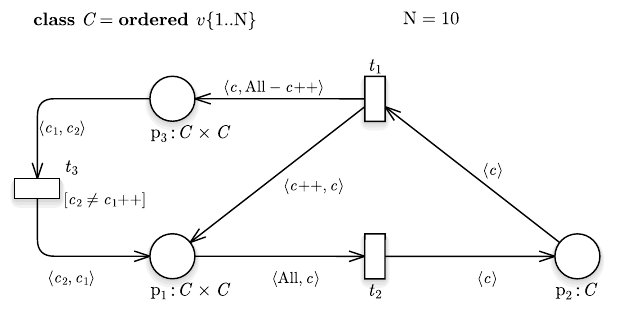}
    \caption{An SN with an ordered class (++ is !)}
    \label{fig:exe2}
\[
\begin{blockarray}{cccc}
 & t_{1_{C}} & t_{2_{C}} & t_{3_{C,C}}  \\
\begin{block}{c[ccc]}
\p{p}_1 & \tuple{!c, c} & -\tuple{All, c} &  \tuple{c_2, c_1}[c_2 \neq \, !c_1] \bigstrut[t] \\
\p{p}_2 & -\tuple{c} & \tuple{c} &   \\
\p{p}_3 & \tuple{c, All - \, !c} &  & -\tuple{c_1, c_2}[c_2 \neq \, !c_1]  \\
\end{block}
\end{blockarray}\vspace*{-1.25\baselineskip}
\]
\[
\begin{blockarray}{cccc}
& \p{p}_1 & \p{p}_2 & \p{p}_3  \\
\begin{block}{c[ccc]}
\vect{I}^1_P  & \tuple{c_2}_{C,C} \quad & \mathrm{N} \tuple{c}_{C} \quad & \tuple{c_1}_{C,C}   \\
\end{block}
\begin{block}{c[ccc]}
\vect{I}^2_P  & \tuple{c_1}_{C,C} &  \tuple{All} & \tuple{c_2}_{C,C} \\
\end{block}
\end{blockarray}\vspace*{-1.3\baselineskip}
\]

\[
\begin{blockarray}{cccc}
& t_1 & t_{2} & t_{3}  \\
\begin{block}{c[ccc]} \vect{I}_T  & \tuple{c}_{C} & \tuple{c}_{C} & \tuple{c, All - \, !c}_{C}  \\
\end{block}
\end{blockarray}\vspace*{-1.3\baselineskip}
\]

\caption{Incidence Matrix and semiflows of SN in Fig. \ref{fig:exe2}}
\label{fig:exe2flows}
\end{figure}

\paragraph{Example 3}
The SN shown in Fig. \ref{fig:exe2} uses an ordered color class $C$ and a color domain in which this class is replicated. The possible symbolic semiflows (Fig. \ref{fig:exe2flows}) correspond to the semiflows of the underlying skeleton, namely $[\p{p}_1:1 \,\, \p{p}_2: |C| \,\, \p{p}_3:1]$ and $[t_1:1 \,\, t_2: 1 \,\, t_3:|C|-1]$. Using the algorithm defined in \cite{pnse2021} verify, for example, that the product of $\vect{I}^1_P$ by the first column of $\vect{H}$ is null. We obtain the following.
\begin{eqnarray*}
   \tuple{c_2} \circ \tuple{!c, c} + N \tuple{c} \circ -\tuple{c} + \tuple{c_1} \circ \tuple{c, All -!c} &\equiv& \tuple{c} - N \tuple{c} + (N-1) \tuple{c} 
\end{eqnarray*}

$\vect{I}^1_P$ means that, in any reachable marking, the colors in the second component of the 2-tuples in $\p{p}_1$ plus those in the first component of the tuples in $\p{p}_3$ always equal $N = |C|$ times the multiset of colors in $\p{p}_2$. Dually, $\vect{I}^2_P$ means that the colors in the first component of the 2-tuples in $\p{p}_1$ and those in the second component of the tuples in $\p{p}_3$ always form the entire color set $C$, with each color occurring $|\vect{m}(\p{p}_2)|$ times.  
The T-semiflow, instead, states that any firing sequence consisting of an occurrence $\tuple{c}$ of $t_1$ and $t_2$, together with all occurrences $\tuple{c, c'}$ of $t_3$ with $c' \neq \, !c$, always returns to the starting marking.

\paragraph{Example 4}
In the small SN shown in Fig. \ref{fig:exe3}, we first encounter arc functions whose images are multisets with multiplicities greater than one. This raises no theoretical issues for our calculus; instead, it is forbidden in \cite{Eva2007}. All possible semiflows can be checked directly. For example, $\vect{I}^1_P$ states that for $\p{p}_1$, the sum of the colors in the two components of its 2-tuples always matches the multiset of colors in $\p{p}_2$.

\begin{figure}[!ht]
    \centering
    \includegraphics[width=0.7\linewidth]{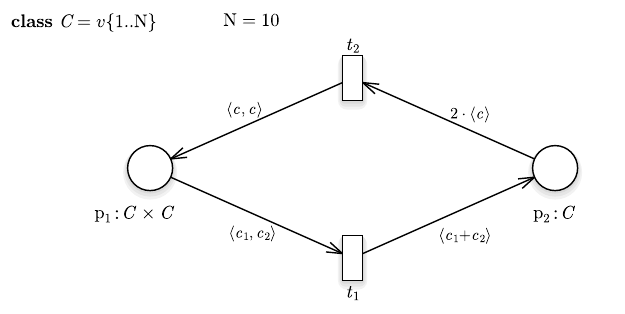}
    \caption{An SN with functions with non-one coefficients}
    \label{fig:exe3}
\[
\begin{blockarray}{ccc}
 & t_{1_{C,C}} & t_{2_{C}}   \\
\begin{block}{c[cc]}
\p{p}_1 & -\tuple{c_1,c_2}_{C,C} & \tuple{c,c}_{C}  \bigstrut[t] \\
\p{p}_2 & \tuple{c_1 + c_2}_{C,C} & -2\tuple{c}_{C}    \\
\end{block}
\end{blockarray}\vspace*{-1.25\baselineskip}
\]
\[
\begin{blockarray}{ccc}
& \p{p}_1 & \p{p}_2   \\
\begin{block}{c[cc]}
\vect{I}_P  & \tuple{c_1 + c_2}_{C,C} &  \tuple{c}_{C} \\
\end{block}
\end{blockarray}
\vspace*{-1.3\baselineskip}
\]
\vspace*{0.2\baselineskip}
\[
\begin{blockarray}{ccc}
& t_1 & t_{2}   \\
\begin{block}{c[cc]}
\vect{I}_T  & \tuple{c,c}_{C} & \tuple{c}_{C}  \\
\end{block}
\end{blockarray}\vspace*{-1.3\baselineskip}
\]

\caption{Incidence Matrix and semiflows of SN in Fig. \ref{fig:exe3}}
\label{fig:exe3flows}
\end{figure}

\section{Generating A Semi-Flow Basis}
\label{sec:basis}
The capability to formally verify symbolic semiflows is particularly valuable, as modelers typically possess a priori knowledge regarding the expected system behavior and seek to corroborate this knowledge through automated verification support. Nevertheless, the capability to construct a semiflow generative family represents an even more critical issue from both theoretical and practical standpoints. As mentioned previously, existing theoretical approaches have achieved only partial success as they rely on imposing rather stringent syntactic constraints on the admissible classes of arc functions.  

In this section, we concisely present a hybrid empirical–theoretical methodology that, instead of relying on advanced algebraic frameworks, exploits Property \ref{propr:skel} under the standing assumption that all arc functions have a constant size.

\begin{claim}
The set of arc functions in $\lang_{\cdom_{\ColClasses},\cdom_{\ColClasses}'}$ of a certain size $k$ is finite and can be computed explicitly.
\end{claim}

The preceding claim is naturally justified by the fact that the collection of class-functions with domain $\cdom_{\ColClasses}$ is uniquely fixed, and the symbolic color-tuples in the domain $\cdom'_{\ColClasses}$ of cardinality $k$ can be generated by an algorithmic enumeration. In general, this enumeration problem, which is theoretically nontrivial, may become computationally expensive when the parameter $k$ is large and-or the color domains are composite. From this point of view, Property \ref{propr:eqT-P-flows} may help. 

We argue that the difficulty of this task can be significantly reduced by leveraging certain properties of SN arc functions. In particular, any arc function $F$ can be expressed in the form $\sum_i \lambda_i T_i$, with $\lambda_i \in \Nat$, where each $T_i$ is a tuple (optionally preceded by a filter) that produces multisets in which all element multiplicities are equal to one.

In this paper, we do not introduce any algorithm; instead, we demonstrate how to use the previous claim to construct a semiflow generative family, drawing on a simple example from Section \ref{sec:semiflow-ver}.
More challenging cases arise when considering the other examples in that section.

Consider the SN depicted in Fig.~\ref{fig:PL}, whose underlying skeleton is provided in the Appendix (Fig.~\ref{fig:skeleton}). The P-semiflow basis of this skeleton consists of the two vectors $[\p{start}:1 \quad \p{line}:1 \quad \p{ass}:1]$ and $[\p{ok}:1 \quad \p{focc}:1]$. In order to derive the corresponding symbolic P-semiflows, we must first select a common subdomain for the places belonging to the support of each semiflow. In this case, there is one candidate subdomain for each of the two P-semiflows, namely $L$ and $PL$, respectively. These classes are neither ordered nor partitioned; consequently, the corresponding admissible class-functions are $\{l,All_L\}$ and $\{pl,All_{PL}\}$. Depending on the place under consideration, these functions may have different domains (for instance, $l_L$ for the place $\p{start}$ and $l_{PL,L}$ for the place $work$); we let the context resolve this ambiguity. Hence, concentrating on the first P-semiflow and observing that the only cardinality-one arc functions that can be constructed are $\tuple{l}$ if $K > 2$ and either $\tuple{l}$ or $\tuple{All - l}$ if $K = 2$, we directly conclude that the symbolic P-semiflow $\vect{I}_p^1$ depicted in Fig.~\ref{fig:PLflows} constitutes part of the basis for $K > 2$. If $K = 2$, the P-vector $[\p{start}:\tuple{All - l}_{L} \quad \p{line}:\tuple{All - l}_{PL,L} \quad \p{ass}:\tuple{All - l}_{L}]$ also constitutes a semiflow, which is equivalent to the previously defined one. We do not elaborate further on this aspect here. The basis of P-semiflows is completed by $\vect{I}_p^2$ (Fig.~\ref{fig:PLflows}), for which analogous comments can be made.

The unique minimal T-semiflow of the SN skeleton is given by  
\[
[load:1 \;\; work:K \;\; assembly:1],
\]  
And therefore, the only admissible subdomain of transition domains is \(PL\).  
This subdomain will serve as the common domain for all functions that make up the symbolic T-semiflow.
With respect to the transition \(work\), we must determine arc functions of cardinality  $K := |L|$. The only feasible candidates for the corresponding semiflow component are $\tuple{pl, All_L}$ and, in the specific case where $N (:= |PL|) = K = 2$,
the additional candidate \(2 \langle All - pl, All_L \rangle\).

\section{Colored Flows}
\label{sec:cflow}

Symbolic flows capture notable properties, both qualitative (color types) and quantitative (color number). In general, however, they comply with rather restrictive, conservative patterns.

It is often advantageous to verify qualitative and quantitative invariants separately. In particular, the preservation of specific structural properties of colored place markings can be independent of the precise number of tokens present, for example, in the case of unbounded models. This is exactly the purpose of \emph{colored semiflows}, hereafter referred to simply as Csemiflows.

Let $\vect{H}^+$ and $\vect{H}^-$ be the $|P|\times|T|$ matrices defined entrywise by
\[
\vect{H}^+[p,t] = \overline{\Out{p,t} \ominus \Inp{p,t}}, 
\qquad
\vect{H}^-[p,t] = \overline{\Inp{p,t} \ominus \Out{p,t}}.
\]
Equivalently, for a given color instance of transition $t$, the entries $\vect{H}^+[p,t]$ and $\vect{H}^-[p,t]$ represent, respectively, the sets of color tuples (tokens) that are added and removed from place $p$.

\begin{definition}[P-Csemiflow] Let $\vect{I}_{P-c}$ be an $1 \times |P|$ vector of functions such that $\vect{I}_{P-c}[p] : \cdom(p) \rightarrow 2^{\cdom'}$ and $\vect{I}_{P-c}[p] \not\equiv 0 \Rightarrow \cdom' \leq \cdom(p)$.
$\vect{I}_{P-c}$ is a P-Csemiflow if and only if $\vect{I}_{P-c} \cdot \vect{H}^+  \,\equiv\, \vect{I}_{P-c} \cdot \vect{H}^-$.
\label{def:P-c-semif}
\end{definition}

\noindent T-Csemiflows are defined analogously. 

\paragraph{Example 5} Consider the SN depicted in Fig. \ref{fig:exe5}.  The underlying skeleton admits one P-semiflow: $[\p{p}_1:2 \quad \p{p}_2: 1 \quad \p{p}_3:3]$; It is straightforward to verify that there are no symbolic semiflows that match it (using property \ref{propr:skel} and considering the functions of proper size). Each transition is equipped with a guard and there is also a self-loop between $t_1$ and $\p{p}_2$. A transition guard is naturally propagated over the incident arc functions. To construct the matrices $\vect{H}^+$ and $\vect{H}^-$, we must rewrite the difference $\vect{H}[\p{p}_2,t_1] := \Out{\p{p}_2,t_1} - \Inp{\p{p}_2,t_1}$ as a pairwise disjoint algebraic sum, which can always be achieved in our framework. In this simple situation, we obtain (omitting the tuple notation; $g_1 \equiv \Guards(t_1)$):
$$(-2 c_1 + c_2)[g_1] \equiv  c_2[c_1 \neq c_2][g_1] -(c_1[c_1 == c_2] +2 c_1[c_1 \neq c_2])[g_1] $$

When considering the support of positive and negative terms, we remove multiplicities and obtain $\vect{H}^+[\p{p}_2,t_1] = c_2[c_1 \neq c_2 \wedge g_1]$, $\vect{H}^-[\p{p}_2,t_1] = c_1[g_1]$. (Hereinafter, in this subsection, we omit the multiset support notation for simplicity.). The other elements of $\vect{H}^+$ and $\vect{H}^-$ are simply derived.

At this stage, we can verify (by applying the calculus in $\overline{\lang}$) whether the two vectors $\vect{I}_P^i$ depicted in Fig.~\ref{fig:exeCflow} actually define colored semiflows. The first would, in that case, correspond to the conservation of the colors in place $\p{p}_2$ and the first component of the 2-tuples in place $\p{p}_1$, respectively, with respect to the subclass $C_1$.  
However, this verification fails. 
Consider, for example, the product of $\vect{I}_P^1$ with the first column of $\vect{H}^+$ and, analogously, with the first column of $\vect{H}^-$:
\begin{eqnarray*}
\tuple{c}[c \in C_1] \circ \tuple{c_2}[g_1 \wedge c_1 \neq c_2]  & \equiv& \tuple{c_2}[g_1 \wedge c_1 \neq c_2] \\
\tuple{c_2}[c_2 \in C_1] \circ \tuple{c_1,c_2}[g_1] + \tuple{c}[c \in C_1] \circ \tuple{c_1}[g_1 ] &\equiv& \tuple{c_2}[g_1 ] + \tuple{c_1}[g_1 ]\\
\end{eqnarray*}
\noindent the resulting expressions, of course, are not equivalent.
In contrast, we easily verify that $\vect{I}_P^2$ is a Csemiflow: the set of colors in place $\p{p}_3$ and the first component of the 2-tuples in $\p{p}_1$ remain invariant.

An interesting concluding observation is that the calculus in $\overline{\lang}$ implemented in \SNex\ is inherently parametric with respect to the cardinalities of color classes, which can be specified by linear constraints; for example, $|C| \geq 3$. Consequently, we are able to validate Csemiflows not only in a purely symbolic manner but also in a parametric setting. For example, the aforementioned results hold irrespective of the cardinality of the subclass \(C_1\).  
A detailed investigation of this parametric analysis is beyond the scope of the present work.

\begin{figure}
    \centering
    \includegraphics[width=0.75\linewidth]{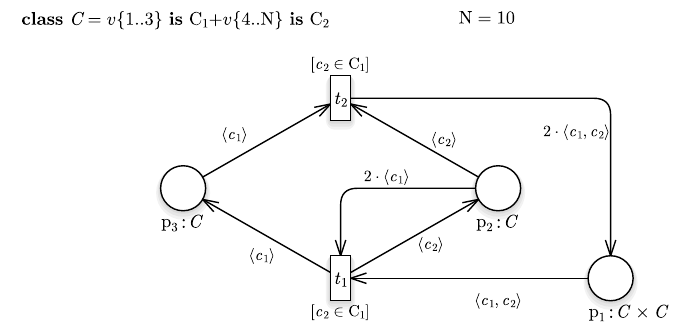}
    \caption{SN with a split color class and guards}
    \label{fig:exe5}
\end{figure}

\begin{figure}[!ht]
\centering
\begin{minipage}{0.48\linewidth}
\centering
\[
\begin{blockarray}{ccc}
\vect{H}^+ & t_{1_{\,C,C}} & t_{2_{\,C,C}}   \\
\begin{block}{c[cc]}
\p{p}_1 &  & \tuple{c_1,c_2}[g_2]  \bigstrut[t] \\
\p{p}_2 & \tuple{c_2}[g_1 \wedge c_1 \neq c_2] &     \\
\p{p}_3 &  \tuple{c_1}[g_1] &     \\
\end{block}
\end{blockarray}
\]
\end{minipage}\hfill
\begin{minipage}{0.48\linewidth}
\centering
\[
\begin{blockarray}{ccc}
\vect{H}^- & t_{\,1_{C,C}} & t_{\,2_{C,C}}   \\
\begin{block}{c[cc]}
\p{p}_1 & \tuple{c_1,c_2}[g_1] & \bigstrut[t] \\
\p{p}_2 & \tuple{c_1}[g_1] & \tuple{c_2}[g_2]    \\
\p{p}_3 &  & \tuple{c_1}[g_2]    \\
\end{block}
\end{blockarray}
\]
\end{minipage}

\[
\begin{blockarray}{cccc}
& \p{p}_1 & \p{p}_2  & \p{p}_3 \\
\begin{block}{c[ccc]}
\vect{I}_P^1  & \tuple{c_2}[c_2 \in C_1]_{C,C} \, \, \, &  \tuple{c}[c \in C_1]_{C} \, \, \, & \\
\end{block}
\begin{block}{c[ccc]}
\vect{I}_P^2  & \tuple{c_1}_{C,C} \, \, \, &  \, \, \, & \tuple{c}_{C}\\
\end{block}
\end{blockarray}
\]

\caption{$\vect{H}^+$ and $\vect{H}^-$ matrices and P-Csemiflow of SN in Fig. \ref{fig:exe5}}
\label{fig:exeCflow}
\end{figure}

\section{Verifying General Invariants}
\label{sec:genInv}
Using a simple yet nontrivial example, we then briefly illustrate how the SN structural calculus implemented in \SNex\ can be employed to rigorously verify properties that go beyond symbolic or colored semiflows. With this method, we broaden the class of models that can be examined (including, for instance, unbounded models) and account for inhibitor arcs, which substantially enhance the expressive power of the SN modeling framework.

The approach is based on a well-established principle. Our goal is to establish the property $\EuScript{P}$, which is formally defined in an extended language $\lang^+$ of $\lang$ (or $\overline{\lang}^+$ if we disregard the quantitative aspects).

\begin{enumerate}
    \item We first show that $\EuScript{P}$ is satisfied in the initial marking (which can also be described symbolically).
    \item We then show that if $\EuScript{P}$ is met in a generic marking $\vect{m}$, then it is also met in any $\vect{m}'$ such that $\vect{m} [(t, c) > \vect{m}'$, for all $t \in T, c \in \cdom(t)$.
\end{enumerate}

Our focus is on the second point, which is the most significant. The essential requirement is that this proof be carried out at a symbolic level, that is, by operating on symbolic transition instances that compactly encode the corresponding concrete instances. To this end, and in analogy with SMT-based techniques, we employ the \SNex\ rewriting engine to infer new axioms (rewriting rules), which are then dynamically incorporated into the underlying theory. Currently, this inference process is performed manually. We plan to extend the \SNex\ CLI to make the overall workflow semi-automated.

From a technical point of view, we must enrich $\lang$ ($\overline{\lang}$) with symbols that denote parametric (multi)sets over specified color domains. We use symbols like $Z_{\cdom_\ColClasses}$ to denote a \emph{parametric} (multi-)set in $\cdom_\ColClasses$, where $\cdom_\ColClasses$ denotes any subdomain made up of color classes in $\ColClasses$.
This symbol can be interpreted as a function with a null (or neutral) domain; that is, $ \bullet \rightarrow Bag[\cdom_\ColClasses]$ or $ \bullet \rightarrow 2^{\cdom_\ColClasses}$, depending on whether we are considering $\EuScript{L}^+$ or $\overline{\EuScript{L}}^+$.

The functional operators of the calculus are then applied in a uniform manner. For example, we may form expressions such as $\tuple{Z_{\cdom_\ColClasses}^1, Z_{\cdom_\ColClasses}^2}$, $Z_{\cdom_\ColClasses}^1 \ \mu\ Z_{\cdom_\ColClasses}^2$, or $F \circ Z_{\cdom_\ColClasses}^1$, where $\mu \in \{+(-),\cap, \ominus \}$, $F \in \fset{\cdom_\ColClasses}{D}$.
    
\begin{figure}[t]
    \centering
    \includegraphics[width=0.7\linewidth]{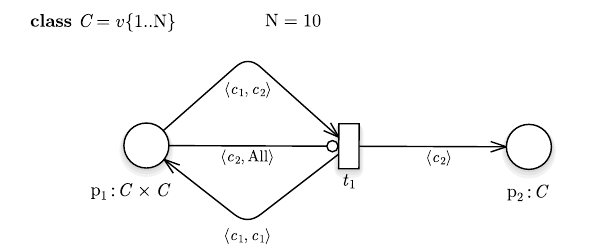}
    \caption{An SN transition with inhibitor edge}
    \label{fig:exe6}
\end{figure}

\paragraph{Example 6}
Let us clarify the concept using the SN in Fig. \ref{fig:exe6}. Our goal is to show that the first component of the 2-tuples (tokens) in place $\p{p}_1$ and the colors in place $\p{p}_2$ are disjoint, under the assumption that this holds in the initial marking $\vect{m}_0$. Demonstrating this directly by inspecting the color functions is not straightforward at all. We can represent an arbitrary SN marking $\vect{m}$ as: $[\vect{m}(\p{p}_1): Z_{C,C}^1 \quad \vect{m}(\p{p}_2): Z_{C}^1]$. The invariant property is (we suppose that '$\circ$' takes priority over'$\cap$' and the latter over '+'):
$${c_1} \circ Z_{C,C}^1 \cap Z_{C}^1 \equiv 0$$

The first natural step consists of incorporating the corresponding rewriting rule ${c_1} \circ Z_{C,C}^1 \cap Z_{C}^1 \longrightarrow 0$ into the framework. Subsequently, we syntactically encode within $\vect{m}$ the enabling condition associated with transition $t_1$. This procedure can, in principle, be fully automated.
A generic symbolic instance of $t_1$ is given by $(t_1,\tuple{Z_{C}^2, Z_{C}^3})$, where $Z_{C}^2$ and $Z_{C}^3$ denote parametric colors in $C$ that are bound to variables (projections) $c_1$ and $c_2$, respectively.  

The constraints imposed by the input and inhibitor arc functions can be represented in the form of rewrite rules. With respect to the input constraint, we obtain:
\[
Z_{C,C}^1 \longrightarrow Z_{C,C}^2 + \tuple{Z_{C}^2, Z_{C}^3}, \quad |Z_{C}^2| = 1, |Z_{C}^3| = 1.
\]

 By substituting this encoding into the invariant expression, which is supposed to hold in $\vect{m}$, and then applying the composition rules together with standard multiset properties—such as the distributivity of composition over multiset sum and the following:
$$
(A + B) \cap C \equiv 0 \Rightarrow A \cap C \equiv 0 \wedge B \cap C \equiv 0,
$$
we finally obtain this expression, which can be simplified:
\begin{eqnarray*}
{c_1} \circ (Z_{C,C}^2 + \tuple{Z_{C}^2, Z_{C}^3}) \cap Z_{C}^1 & \equiv  {c_1} \circ Z_{C,C}^2 \cap Z_{C}^1 + Z_{C}^2 \cap Z_{C}^1  \equiv \emptyset \\
\end{eqnarray*}
\noindent From which we infer the rules: i) ${c_1} \circ Z_{C,C}^2 \cap Z_{C}^1\longrightarrow 0$, ii) $Z_{C}^2 \cap Z_{C}^1  \longrightarrow 0$.

The constraint imposed by the inhibitor arc function is formalized as follows:
\begin{eqnarray*}
\tuple{Z_{C}^3, All} \cap (Z_{C,C}^2 + \tuple{Z_{C}^2, Z_{C}^3}) < \tuple{All, All} & \cong \\
\tuple{Z_{C}^3, All} \cap (Z_{C,C}^2 + \tuple{Z_{C}^2, Z_{C}^3})  \equiv 0 & \cong \\
\tuple{Z_{C}^3, All} \cap Z_{C,C}^2 \equiv 0 \, \wedge \, \tuple{Z_{C}^3, All} \cap \tuple{Z_{C}^2, Z_{C}^3} \equiv 0 &
\end{eqnarray*}
\noindent From which we infer the rules: iii) $Z_{C}^3 \cap c_1 \circ Z_{C,C}^2 \longrightarrow 0$,  iv) $Z_{C}^3 \cap Z_{C}^2 \longrightarrow 0$.

\vspace{5pt}
The marking expression we obtain by symbolically firing the instance $(t_1,\tuple{Z_{C}^2, Z_{C}^3})$ in \vect{m} is: 

\vspace{5pt}
\hspace{2.5cm}$[\vect{m}'(\p{p}_1): Z_{C,C}^2 + \tuple{Z_{C}^2,Z_{C}^2} \quad \vect{m}'(\p{p}_2): Z_{C}^1 + Z_{C}^3]$.

\vspace{5pt}
\noindent Finally, it remains to establish the following result:
\begin{align*}
  {c_1} \circ (Z_{C,C}^2 + \tuple{Z_{C}^2,Z_{C}^2}) \cap (Z_{C}^1 + Z_{C}^3)  \equiv 
  ({c_1} \circ Z_{C,C}^2 + Z_{C}^2) \cap (Z_{C}^1 + Z_{C}^3)  \equiv 0 & \cong  \\
  {c_1} \circ (Z_{C,C}^2) \cap Z_{C}^1  \equiv 0 \wedge {c_1} \circ (Z_{C,C}^2) \cap Z_{C}^3  \equiv 0 \wedge
  Z_{C}^2 \cap  Z_{C}^1  \equiv 0 \wedge
  Z_{C}^2 \cap  Z_{C}^3  \equiv 0 &   
\end{align*}

\noindent which can be directly inferred from axioms i)-iv).

As a final remark, we note that this methodology can also be employed to validate candidate symbolic P-semiflows (or colored flows). However, in this case, the validation is achieved through a more involved and less streamlined procedure than the approaches presented in Sections \ref{sec:semiflow-ver} and \ref{sec:cflow}.

\section{Conclusions} We have shown that the symbolic structural calculus for Symmetric Nets (SN)—developed over the last twenty years to infer structural dependencies among SN nodes and realized in the \SNex\ tool—can likewise be used to validate a wide range of symbolic (semi)flows for extended Symmetric Net (ESN) models. In addition, we have carried out an initial study on deriving a semiflow basis by leveraging a natural link to the standard semiflows of the ESN skeleton. Lastly, we have indicated how this calculus could be applied to formally check more general categories of structural invariants.


In practice, we intend to enhance \SNex\ by (1) providing full support for composing multiset functions, (2) introducing an interactive inference facility that supplies newly inferred axioms to the \SNex\ rewriting engine, and (3) ensuring seamless interoperability with other tools, in particular \GreatSPN.  
On the theoretical side, we investigate whether generalized inverses of arc functions (linearly extended) can be represented symbolically, so as to build a generative family of symbolic flows, possibly leveraging the preliminary findings reported in this paper.


  

\bibliographystyle{eptcs}
\bibliography{biblio}

\newpage
\newpage

\appendix

\section*{Appendix: proofs, unfolding-skeletons, and conventional semiflows}

\proof{Property \ref{propr:skel}. Consider the $j^{th}$ column of $\vect{H}$. The row-by-column product $\sum_{p \in P} \vect{I}_P[p] \circ \vect{H}[p,t_j]$ results in the null function (Hp). Since the function linear extension is also constant-size, $\vect{I}_P[p] \circ \vect{H}[p,t_j]$ is constant-size and $size(\vect{I}_P[p] \circ \vect{H}[p,t_j]) = size(\vect{I}_P[p]) \cdot (size(\Out{p,t_j}) - size(\Inp{p,t_j})$.\\ Therefore, the algebraic sum corresponding to the conventional row-by-column product is zero; thus, $\vect{I}'_P$ is a conventional flow of $\EuScript{N}_{skel}$.

\noindent Property \ref{propr:eqT-P-flows}. It follows directly from the definition of the transpose of a matrix, together with the fundamental rule connecting the transpose of a function with composition:
$g \circ f \equiv f^t \circ {g}^t$    
}


   
\begin{figure}[h]
\begin{center}
\includegraphics[width=0.8\textwidth]{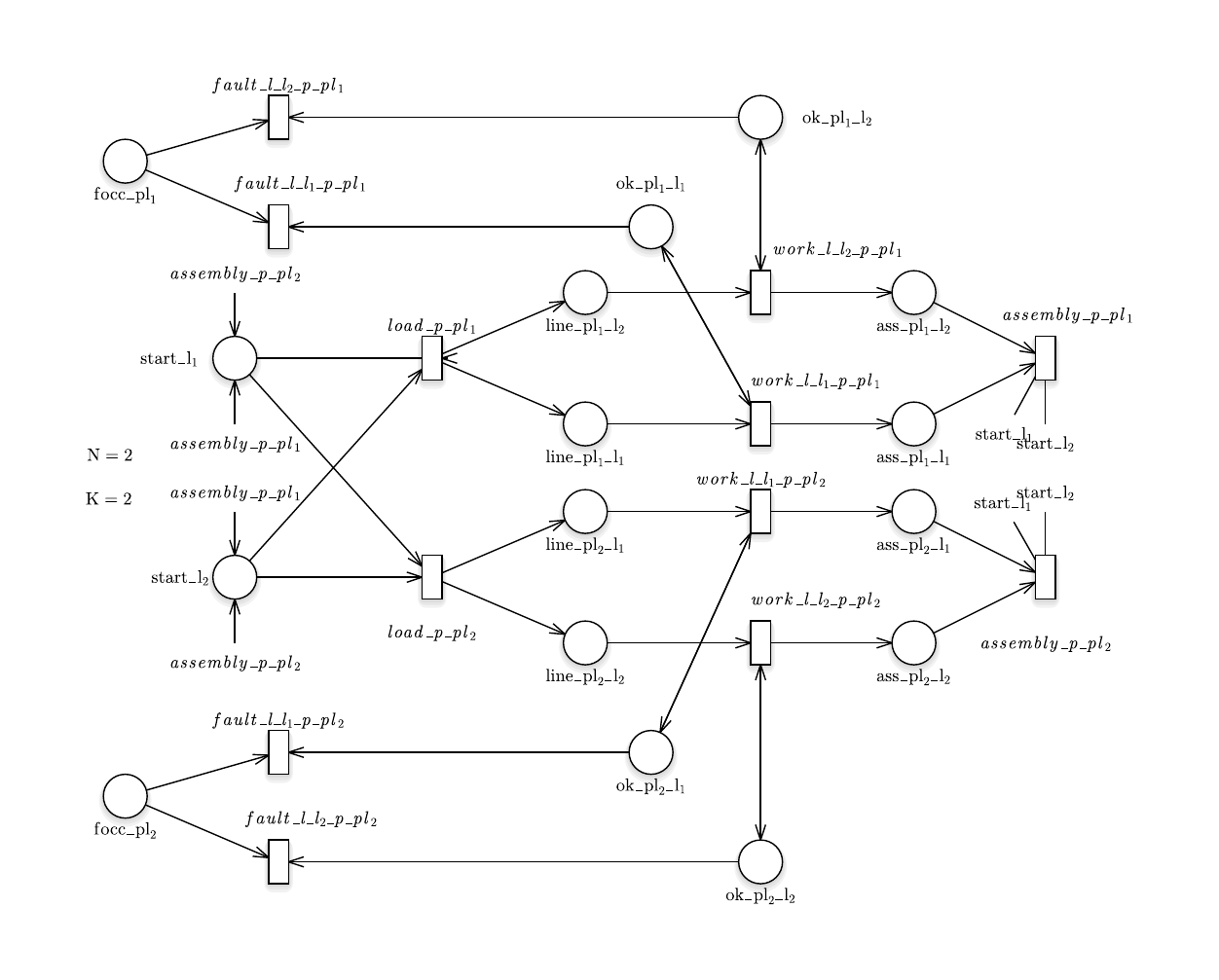}

\caption{Unfolding of the PL model ($N = 2$, $K = 2$)}
\label{fig:unfolding}
\end{center}
\end{figure}

\begin{center}
{\footnotesize
\begin{tabular}{lllll}
\hline
 \p{start\_l_1} & \p{line\_pl_1\_l_2} & \p{line\_pl_2\_l_2} & \p{ass\_pl_1\_l_2} & \p{ass\_pl_2\_l_2} \\
 \p{start\_l_1} & \p{line\_pl_1\_l_1} & \p{line\_pl_2\_l_2} & \p{ass\_pl_1\_l_1} & \p{ass\_pl_2\_l_2} \\
 \p{start\_l_1} & \p{line\_pl_1\_l_2} & \p{line\_pl_2\_l_1} & \p{ass\_pl_1\_l_2} & \p{ass\_pl_2\_l_1} \\
 \p{start\_l_1} & \p{line\_pl_1\_l_1} & \p{line\_pl_2\_l_1} & \p{ass\_pl_1\_l_1} & \p{ass\_pl_2\_l_1} \\
 \p{start\_l_2} & \p{line\_pl_1\_l_2} & \p{line\_pl_2\_l_2} & \p{ass\_pl_1\_l_2} & \p{ass\_pl_2\_l_2} \\
 \p{start\_l_2} & \p{line\_pl_1\_l_1} & \p{line\_pl_2\_l_2} & \p{ass\_pl_1\_l_1} & \p{ass\_pl_2\_l_2} \\
 \p{start\_l_2} & \p{line\_pl_1\_l_2} & \p{line\_pl_2\_l_1} & \p{ass\_pl_1\_l_2} & \p{ass\_pl_2\_l_1} \\
 \p{start\_l_2} & \p{line\_pl_1\_l_1} & \p{line\_pl_2\_l_1} & \p{ass\_pl_1\_l_1} & \p{ass\_pl_2\_l_1}\\
\end{tabular}
\vspace{1em}

\begin{tabular}{llll}
\hline
$\mathit{load\_p\_pl_2}$ & $\mathit{work\_l\_l_1\_p\_pl_2}$ & $\mathit{work\_l\_l_2\_p\_pl_2}$ & $\mathit{assembly\_p\_pl_2}$ \\
$\mathit{load\_p\_pl_1}$ & $\mathit{work\_l\_l_1\_p\_pl_1}$ & $\mathit{work\_l\_l_2\_p\_pl_1}$ & $\mathit{assembly\_p\_pl_1}$ \\
\hline
\end{tabular}

\captionof{table}{P- and T-semiflows of the Petri net shown in Fig. \ref{fig:unfolding}.}
}
\end{center}

\begin{figure}[h]
\begin{center}
\includegraphics[width=0.7\textwidth]{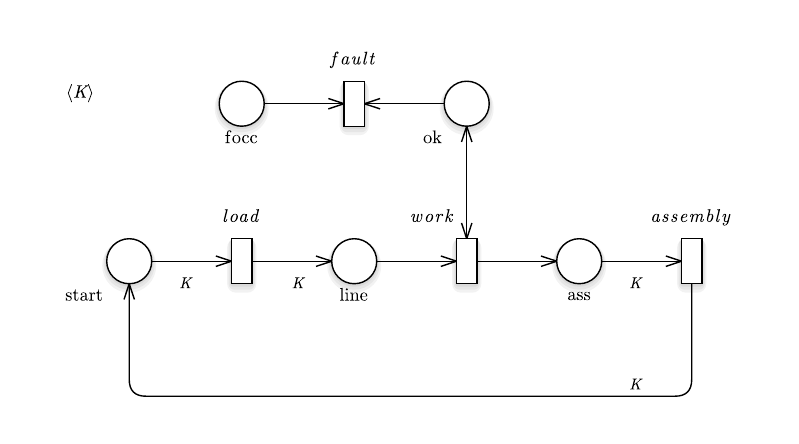}
\caption{Skeleton of the PL model}
\label{fig:skeleton}

\includegraphics[width=0.9\textwidth]{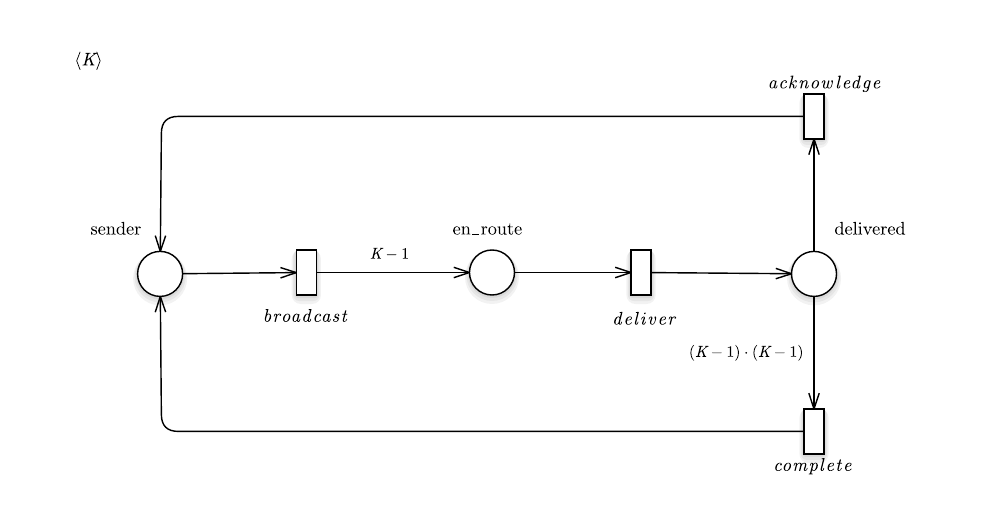}
\caption{Skeleton of the broadcast model}
\label{fig:skeletonbroad}
\end{center}
\end{figure}

\end{document}